\documentclass{article}

\usepackage{amsmath,amsfonts,amssymb,amsthm}
\usepackage[round]{natbib}
\usepackage{mathtools} 
\usepackage{color}
\usepackage{comment}
\usepackage{multirow}
\usepackage{verbatim}
\usepackage{bbm}
\usepackage{booktabs}
\usepackage{float}
\usepackage{graphicx}
\usepackage{paralist} 
\usepackage{etoolbox}
\usepackage[margin=1.5in]{geometry}

\allowdisplaybreaks

\newcommand{\wt}{\widetilde}

\AtBeginEnvironment{quote}{\singlespacing}

\newcommand{\bsx}{\boldsymbol{x}}

\newcommand{\bsX}{\boldsymbol{X}}

\newcommand{\E}{\mathbb{E}}

\newcommand{\Var}{\mathrm{Var}}

\newcommand{\odb}{\mathcal{O}}

\newcommand{\bern}{\mathrm{Bern}}

\newcommand{\phm}{\phantom{-}}

\newcommand{\sd}{\mathrm{SD}}

\newcommand{\bstau}{\boldsymbol{\tau}}

\newcommand{\htaur}{\boldsymbol{\hat \tau_r}}

\newtheorem{proposition}{Proposition}

\newtheorem{lemma}{Lemma}

\theoremstyle{definition} 
\newtheorem{assumption}{Assumption}
\newcommand{\diag}{\mathrm{diag}}
\newcommand{\OR}{\text{OR}}
\newcommand{\Conv}{\text{Conv}}

\newcommand{\bsv}{\boldsymbol{v}}

\begin{document}

\title{Designing Experiments Informed \\by Observational Studies}
\author{Evan Rosenman, Art Owen}
\maketitle

\abstract{

The increasing availability of passively observed data has yielded a growing methodological interest in ``data fusion." These methods involve merging data from observational and experimental sources to draw causal conclusions -- and they typically require a precarious tradeoff between the unknown bias in the observational dataset and the often-large variance in the experimental dataset. We propose  an alternative approach to leveraging observational data, which avoids this tradeoff: rather than using observational data for inference, we use it to design a more efficient experiment. 

We consider the case of a stratified experiment with a binary outcome, and suppose pilot estimates for the stratum potential outcome variances can be obtained from the observational study. We extend results from \cite{zhao2019sensitivity} in order to generate confidence sets for these variances, while accounting for the possibility of unmeasured confounding. Then, we pose the experimental design problem as one of regret minimization, subject to the constraints imposed by our confidence sets. We show that this problem can be converted into a convex minimization and solved using conventional methods. Lastly, we demonstrate the practical utility of our methods using data from the Women's Health Initiative. 

\tableofcontents 
\section{Introduction}

The past half-century of causal inference research has engendered a healthy skepticism toward observational data \citep{Imbens:2015:CIS:2764565}. In observational data sets, researchers do not control whether or not each individual receives a treatment of interest. Hence, they cannot be certain that treated individuals and untreated individuals are otherwise comparable. 

This challenge can be overcome only if the covariates measured in the observational data are sufficiently rich to fully explain who receives the treatment and who does not. This is a fundamentally untestable assumption -- and even if it holds, careful modeling is necessary to remove the selection effect. The applied literature includes myriad examples of treatments that showed promise in observational studies only to be overturned by later randomized trials \citep{hartman2015sate}. One prominent case, the effect of hormone therapy on the health of postmenopausal women, will be discussed in this manuscript \citep{writing2002risks}.  


The ``virtuous" counterpart to observational data is the well-designed experiment. Data from a randomized trial yield unbiased estimates of a causal effect without the need for problematic statistical assumptions. Yet experiments are frequently expensive, and, as a consequence, generally involve fewer units. Especially if one is interested in subgroup causal effects, this means experimental estimates can be imprecise. 

In this paper, we discuss an approach that allows us to leverage the  availability of observational data, while retaining the attractive unbiasedness properties of randomized experiments: we use the observational data not for inference, but rather to influence the design of the experiment. Our discussion will be limited to settings with binary outcomes, in which computations are tractable. We suppose the experiment has a stratified design, and seek to determine allocations of units to strata and treatment assignments. 

Suppose pilot estimates of the stratum potential outcome variances are obtained from the observational study. If the outcomes are binary, we show that recent advances in sensitivity analysis from \cite{zhao2019sensitivity} can be extended to generate confidence sets for these variances, while incorporating the possibility of unmeasured confounding. Next, we pose the experimental design problem as one of regret minimization subject to the potential outcome variances lying within their confidence sets. We use a trick from von Neumann to convert the problem into a convex (though non-DCP) minimization, which can be solved using projected gradient descent. This approach can yield modest efficiency gains in the experiment, especially if there is heterogeneity in treatment effects and baseline incidence rates across strata. 


The remainder of the paper proceeds as follows. Section \ref{sec:problemSetup} defines our notation, assumptions, and loss function. Section \ref{sec:results} gives our main results. These include the derivation of bias-aware confidence sets for the pilot variance estimates; the formulation of the design problem as a regret minimization; and the strategy to convert that problem into a computationally tractable one. We demonstrate the practical utility of our methods on data from the Women's Health Initiative in Section \ref{sec:whi}. Section \ref{sec:conc} discusses future work and concludes. 

\section{Problem Set-Up}\label{sec:problemSetup}

\subsection{Sources of Randomness}

We suppose we have access to an observational study with units $i$ in indexing set $\mathcal{O}$ such that $|\mathcal{O}| = n_o$. We associate with each unit $i \in \mathcal{O}$ a pair of unseen potential outcomes $(Y_i(0), Y_i(1))$; an observed covariate vector $X_i$ where $X_i \in \mathbb{R}^p$; a propensity score $p_i \in (0, 1)$ denoting that probability of receiving treatment. We also associate with each $i$ a treatment indicator $W_i$ and an observed outcome defined by $Y_i = W_i Y_i(1) + (1-W_i) Y_i(0)$. 

There are multiple perspectives on randomness in causal inference. In the setting of \cite{rubin1974estimating} -- as in much of the early potential outcomes literature -- all quantities are treated as fixed except, the treatment assignment $W_i$. More modern approaches sometimes treat the potential outcomes $Y_i(0)$ and $Y_i(1)$ and covariates $X_i$ as random variables \citep[see e.g.][]{vanderweele2012stochastic}. Similarly, some authors treat all of the data elements (including the treatment assignment $W_i$) as random draws from a super-population \citep[see e.g.][]{Imbens:2015:CIS:2764565}. Per the discussion in \cite{chin2019modern}, these subtleties often have little effect on the choice of estimators, but they do affect the population to which results can be generalized. 

In our setting, we assume that the RCT data has not yet been collected, so it does not make sense to talk about their fixed potential outcomes. More naturally, we treat the potential outcomes and covariates as random. Thus, we view $X_i, Y_i(0), Y_i(1)$ as drawn from a joint distribution $F_O$. The RCT data will be denoted (with a slight abuse of notation) as $(Y_i(0), Y_i(1), X_i)$ for $i \in \mathcal{R}$, sampled from a joint distribution $F_R$. Because we are treating the potential outcomes as random variables, we can reason about their means and variances under the distribution $F_R$. 

\subsection{Stratification and Assumptions}

We will make the following assumptions about allocation to treatment.

\setcounter{assumption}{0}

\begin{assumption}[Allocations to Treatment]
For $i \in \mathcal{O}$, $W_i \sim \bern(p_i)$ for $p_i$. For $i \in \mathcal{R}$, treatment is allocated via a simple random sample of size $n_{rkt}$ for $k = 1, \dots, K$.
\end{assumption}

We suppose we have a fixed stratification scheme based on the covariates $X_i$. This can be derived from substantive knowledge or from applying a modern machine learning algorithm on the observational study to uncover treatment heterogeneity \citep[e.g.][]{wager2018estimation, hill2011bayesian}. The stratification is such that there are $k = 1, \dots, K$ strata and each has an associated population weight $w_1, \dots, w_K$.  Using the stratification on the observational study, we define indexing subsets $\mathcal{O}_k$ with cardinalities $n_{ok}$ to identify units in each stratum. For each stratum, define $\mathcal{I}_k$ as the set of covariate values defining the stratum, such that $X_i \in \mathcal{I}_k \implies i \in \mathcal{O}_k$.

Suppose we can recruit only $n_r$ total units for the RCT. We need to decide both the number of units $n_{rk}$ recruited for each stratum, subject to the constraint $\sum_k n_{rk} = n_r$, and the count of units we will assign to treatment vs. control in each stratum, such that the associated counts $n_{rkt}$ and $n_{rkc}$ sum to $n_{rk}$. Hence, our variables of interest will be $\{(n_{rkt}, n_{rkc})\}_1^K$.

Define $\E_R, \Var_R, \E_O,$ and $\Var_O$ as expectations and variances under the distributions $F_R$ and $F_O$, respectively. We will need two further assumptions. 

\begin{assumption}[Common Potential Outcome Means]\label{ass:commonMeans}
Conditional on the stratum, the potential outcome averages for the two populations are equal. In other words,
\begin{align*}
\E_R(Y_i(0) \mid X_i \in \mathcal{I}_k) &= \E_O(Y_i(0) \mid X_i \in \mathcal{I}_k) \hspace{4mm} \text{ and }  \\
\E_R(Y_i(1) \mid X_i \in \mathcal{I}_k) &= \E_O(Y_i(1) \mid X_i \in \mathcal{I}_k) 
\end{align*}
for all $k \in 1, \dots, K$. We denote these shared quantities as $\mu_k(0)$ and $\mu_k(1)$. 
\end{assumption}

\begin{assumption}[Common Potential Outcome Variances]\label{ass:commonVars}
Conditional on the stratum, the potential outcome means for the two populations are equal. In other words,
\begin{align*}
\Var_R(Y_i(0) \mid X_i \in \mathcal{I}_k) &= \Var_O(Y_i(0) \mid X_i \in \mathcal{I}_k) \hspace{4mm} \text{ and }  \hspace{4mm} \\\Var_R(Y_i(1) \mid X_i \in \mathcal{I}_k) &= \Var_O(Y_i(1) \mid X_i \in \mathcal{I}_k)
\end{align*} 
for all $k \in 1, \dots, K$. We denote these shared quantities as $\sigma_k^2(0)$ and $\sigma_k^2(1)$. 
\end{assumption}

\subsection{Loss and Problem Statement}

Given Assumption \ref{ass:commonMeans}, we can define a mean effect,
\[ \tau_k = \E_R(Y_i(1) - Y_i(0) \mid X_i \in \mathcal{I}_k) = \E_O(Y_i(1) - Y_i(0) \mid X_i \in \mathcal{I}_k) = \mu_k(1) - \mu_k(0) \] 
for each $k \in 1, \dots, K$. We can collect these values into a vector $\bstau$. 

Denote the associated causal estimates derived from the RCT as $\hat \tau_{rk}$ for $k = 1, \dots, K$. We can collect these estimates into a vector $\htaur$. We use a weighted $L_2$ loss when estimating the causal effects across strata, 

\[ \mathcal{L}(\bstau, \htaur) =  \sum_k w_k \left( \hat \tau_k - \tau_k \right)^2 \,.\]
 
Our goal will be to minimize the risk, defined as an expectation of the loss over both the treatment assignments and the potential outcomes. For simplicity, we suppress the subscript and write
\begin{align*}
\mathit{R}(\bstau, \htaur) &=  \E \left(  \sum_k w_k \left( \hat \tau_k - \tau_k \right)^2 \right) \\
&=  \sum_k w_k \left(\frac{\sigma_{k}^2(1)}{n_{rkt}} +\frac{\sigma_{k}^2(0)}{n_{rkc}}\right) \,.
\end{align*}

\section{Converting to an Optimization Problem}\label{sec:results}

\subsection{Na\"ive Approach}

Were $\left(\sigma_k^2(1), \sigma_k^2(0)\right)_{k = 1}^K$ known exactly, it would be straightforward to compute optimal allocations in the RCT.
The optimal choice from minimizing this quantity is simply: 
\begin{equation}\label{eq:alloc}
n_{rkt} = n _r \frac{\sqrt{w_k} \sigma_k(1)}{\sum_j \sqrt{w_k} (\sigma_j(1) + \sigma_j(0)) }, \hspace{5mm} n_{rkc} = n _r \frac{\sqrt{w_k} \sigma_k(0)}{\sum_j \sqrt{w_k} (\sigma_j(1) + \sigma_j(0)) }
\end{equation} 
which yields a risk of 
\[ \frac{1}{n_r} \left( \sum_k \sqrt{w_k} \bigg(\sigma_k(1) + \sigma_k(0)\bigg) \right)^2 \,. \] 

Assumption \ref{ass:commonVars} guarantees shared variance across the observational and RCT datasets. So we might be tempted to obtain pilot estimates of $\sigma_k^2(1)$ and $\sigma_k^2(0)$ from the observational study and then to plug them in to determine the allocation of units in the RCT. However, any estimate of the variances derived from the observational study should be treated with caution. Our assumptions do not preclude the possibility of unmeasured confounding, which can introduce substantial bias into the pilot estimation step. Hence, a framework that exclusively optimizes expected loss is incongruent with what we know about sources of uncertainty. 

\subsection{Regret Minimization}

Decision theory provides an attractive framework in the form of regret minimization, originally attributed to \cite{bell1982regret}, as well as \cite{loomes1982regret}. In this framework, a decision-maker chooses between multiple prospects, and cares not only about the received payoff but also about the foregone choice. If the foregone choice would have yielded higher payoff than the chosen one, the decision-maker experiences regret \citep{diecidue2017regret}. Decisions are made to minimize the maximum possible regret. 

In our case, the decision is in how to allocate units in our RCT. One choice is an allocation informed by the observational study. The other is a ``default" allocation against which we seek to compare. Denote the default values as $\tilde n_{rkt}$ and $\tilde n_{rkc}$, where a common choice would be equal allocation, $\tilde n_{rkt} = \tilde n_{rkc} = n_r/2K$ for all $k$; or weighted allocation $\tilde n_{rkt} = \tilde n_{rkc} = w_k n_r$ for all $k$. 

Regret is defined as the difference between the risk of our chosen allocation and the default allocation,
\[ \text{Regret}\left( \{n_{rkt}, n_{rkc}\}_{k = 1}^K \right) = \sum_k w_k \left(\sigma_k^2(1) \left(\frac{1}{n_{rkt}}- \frac{1}{\tilde n_{rkt}}\right) + \sigma_k^2(0) \left(\frac{1}{n_{rkc}}- \frac{1}{\tilde n_{rkc}}\right)\right) \,.\] 
Choosing this as our objective, we can now begin to formulate an optimization problem. 

Suppose we can capture our uncertainty about $(\sigma_k^2(1), \sigma_k^2(0))$ via a convex constraint, indexed by a user-defined parameter $\Gamma$,
\[ (\sigma_k^2(1), \sigma_k^2(0)) \in \mathcal{A}_k^{(\Gamma)}, \hspace{5mm} k = 1, \dots, K\,, \] 
where $ \mathcal{A}_k^{(\Gamma)} \subset \mathbb{R}^2$. We could then obtain the regret-minimizing unit allocations as the solution to
\begin{equation}\label{optProb}
\begin{aligned}
\min_{n_{rkt}, n_{rkc}} \max_{\sigma_{k}^2(1), \sigma_{k}^2(0)} & \hspace{5mm}  \sum_k w_k \left(\sigma_k^2(1) \left(\frac{1}{n_{rkt}}- \frac{1}{\tilde n_{rkt}}\right) + \sigma_k^2(0) \left(\frac{1}{n_{rkc}}- \frac{1}{\tilde n_{rkc}}\right)\right) \\ \text{subject to} & \hspace{5mm} (\sigma_k^2(1), \sigma_k^2(0)) \in \mathcal{A}_k^{(\Gamma)},  k = 1, \dots, K\\ & \hspace{5mm} \sum_k n_{rkt} + n_{rkc} = n_r \,.
\end{aligned}
\end{equation}
Defining and solving Optimization Problem \ref{optProb} will be the goal of the remainder of this paper.

\subsection{Tractable Case: Binary Outcomes}\label{subsec:tractable}

To construct our confidence regions $\mathcal{A}_k, k = 1, \dots, K$, we will extend recent sensitivity analysis results from  \cite{zhao2019sensitivity}. 

The authors consider the case of causal estimation via stabilized inverse probability of treatment weighting (SIPW). Zhao and co-authors focus on observational studies, and consider the case where unmeasured confounding is present. To quantify this confounding, they propose a marginal sensitivity model indexed by a quantity $\Gamma$, which bounds the odds ratio between the true treatment probability (a function of the covariates and the potential outcomes) and the treatment probability marginalized over the potential outcomes (a function of the covariates only). Their method extends the widely-used Rosenbaum sensitivity model \citep{rosenbaum1987sensitivity}. 

The authors' focus is on developing valid confidence intervals for the average treatment effect even when $\Gamma$-level confounding may be present. They offer two key insights. First, they demonstrate that for any choice of $\Gamma$, one can efficiently compute upper and lower bounds on the true potential outcome means via linear fractional programming. These bounds, referred to as the ``partially identified region," quantify the possible bias in the point estimate of the ATE. Second, the authors show that the bootstrap is valid in this setting. Hence, they propose drawing repeated bootstrap replicates; computing extrema within each replicate using their linear fractional programming approach; and then taking the relevant $\alpha$-level quantiles of these extrema. This procedure yields a valid $\alpha$-level confidence region for the ATE. 

We adapt this approach to our setting in the case of binary outcomes. Note that if $Y_i \in \{0, 1\}$, then potential outcome variances can be expressed directly as a function of potential outcome means, via 
\[ \sigma_k^2(1) = \mu_k(1)\cdot (1 - \mu_k(1)) \hspace{3mm} \text{ and } \hspace{3mm} \sigma_k^2(0) = \mu_k(0) \cdot(1-\mu_k(0))\,. \] 

As \cite{zhao2019sensitivity} provides the necessary machinery to bound mean estimates, we can exploit this relationship between the means and variances to bound variance estimates. In particular, we can show that the bootstrap is also valid if our estimand is $\mu_k(e)\cdot(1 - \mu_k(e))$, rather than $\mu_k(e)$, for $e \in \{0, 1\}$ and $k = 1, \dots K$. Computing the extrema is also straightforward. Note that the function $f(x) = x\cdot(1-x)$ is monotonically increasing in $x$ if $0 < x < 0.5$ and monotonically decreasing in $x$ if $0.5 < x < 1$. Hence, if we use the \cite{zhao2019sensitivity} method to solve for a partially identified region for $\mu_k(1)$ and $\mu_k(0)$, we can equivalently compute such intervals for $\sigma_k^2(1)$ and $\sigma_k^2(0)$. 

Denote as $\hat \mu_k^U(e)$ the upper bound and $\hat \mu_k^L(e)$ the lower bound computed for a mean for $e \in \{0, 1\}$. Denote $(\hat \sigma_k^2(e))^U$ and $(\hat \sigma_k^2(e))^L$ as the analogous quantities for variance. We apply the following logic:
\begin{itemize}
    \item If $\hat\mu_k^U(e) \leq 0.5$, set 
 \[ (\hat \sigma_k^2(e))^L = \mu_k^L(e)(1-\hat \mu_k^L(e)) \hspace{3mm} \text{ and } \hspace{3mm} (\hat \sigma_k^2(e))^U = \mu_k^U(e)(1-\hat \mu_k^U(e))\,. \] 
    \item If $\hat\mu_k^L(e) \geq 0.5$, set 
 \[ (\hat \sigma_k^2(e))^L = \mu_k^U(e)(1-\hat \mu_k^U(e)) \hspace{3mm} \text{ and } \hspace{3mm} (\hat \sigma_k^2(e))^U = \mu_k^L(e)(1-\hat \mu_k^L(e)) \,. \] 
    \item If $\hat\mu_k^L(e) < 0.5$ and $\hat\mu_k^U(e) > 0.5$, set
    \[ (\hat \sigma_k^2(e))^L = \min\left( \hat \mu_k^L(e)(1-\hat \mu_k^L(e)), \hat \mu_k^U(e)(1-\hat \mu_k^U(e))\right) \hspace{3mm} \text{ and } \hspace{3mm} (\hat \sigma_k^2(e))^U = 0.25 \,.\]
\end{itemize}

Hence, we propose the following procedure for deriving valid confidence regions for $(\sigma_k^2(0), \sigma_k^2(1))$ for each choice of $k$: 
\begin{enumerate}
\item Draw $B$ bootstrap replicates from the units $i \in \mathcal{O}_k$. 
\item For each replicate: 
\begin{itemize}
\item Compute $\mu_k^U(e), \mu_k^L(e)$ for $e \in \{0, 1\}$ using Zhao and co-authors' linear fractional programming approach. 
\item Determine $ (\hat \sigma_k^2(e))^U$ and $ (\hat \sigma_k^2(e))^L$ for $e \in \{0, 1\}$ using the approach described above. 
\end{itemize}
\item Each replicate can now be represented as a rectangle in $[0, 1] \times [0, 1]$, where one axis represents the value of $(\hat \sigma_k^2(1))$, and the other the value of $(\hat \sigma_k^2(0))$ and the vertices correspond to the extrema. Any set such that a $1-\alpha$ proportion of the rectangles have all four corners included in the set will asymptotically form a valid $\alpha$-level confidence interval. 
\end{enumerate}
A full proof of the validity of this method can be found in Appendix \ref{app:validConfRegions}.

Note that the final step does not specify the shape of the confidence set (it need not even be convex). For simplicity, we compute the minimum volume ellipsoid containing all vertices, then shrink the ellipsoid toward its center until only $B\cdot(1-\alpha)$ of the rectangles have all four of their vertices included. For details on constructing the ellipsoids (sometimes known as L\"{o}wner-John ellipsoids), see \cite{boyd2004convex}. Observe that this is by no means the smallest valid confidence set, but it is convex and easy to work with numerically. 

In Figure \ref{fig:rectplot}, we demonstrate this procedure on simulated data using $\Gamma = 1.2$. We suppose there are four strata, each containing 1,000 observational units. The strata differ in their treatment probabilities with 263, 421, 564, and 739 units in each stratum, respectively. The large black dot at the center of each cluster represents the point estimate $(\hat \sigma_k^2(0), \hat \sigma_k^2(1))$. In purple, we plot the rectangles corresponding to the extrema computed in each of 200 bootstrap replicates drawn from the data. The dashed ellipsoids represent 90\% confidence sets. In the cases of strata 2 and 4, the ellipsoids extend beyond the upper bound of 0.25 in at least one direction, so we intersect the ellipsoids with the hard boundary at 0.25. The resulting final confidence sets, $\mathcal{A}_1, \mathcal{A}_2, \mathcal{A}_3,$ and $\mathcal{A}_4$,  are all convex. 

\begin{figure}[h]
\centering
\includegraphics[scale = 0.30]{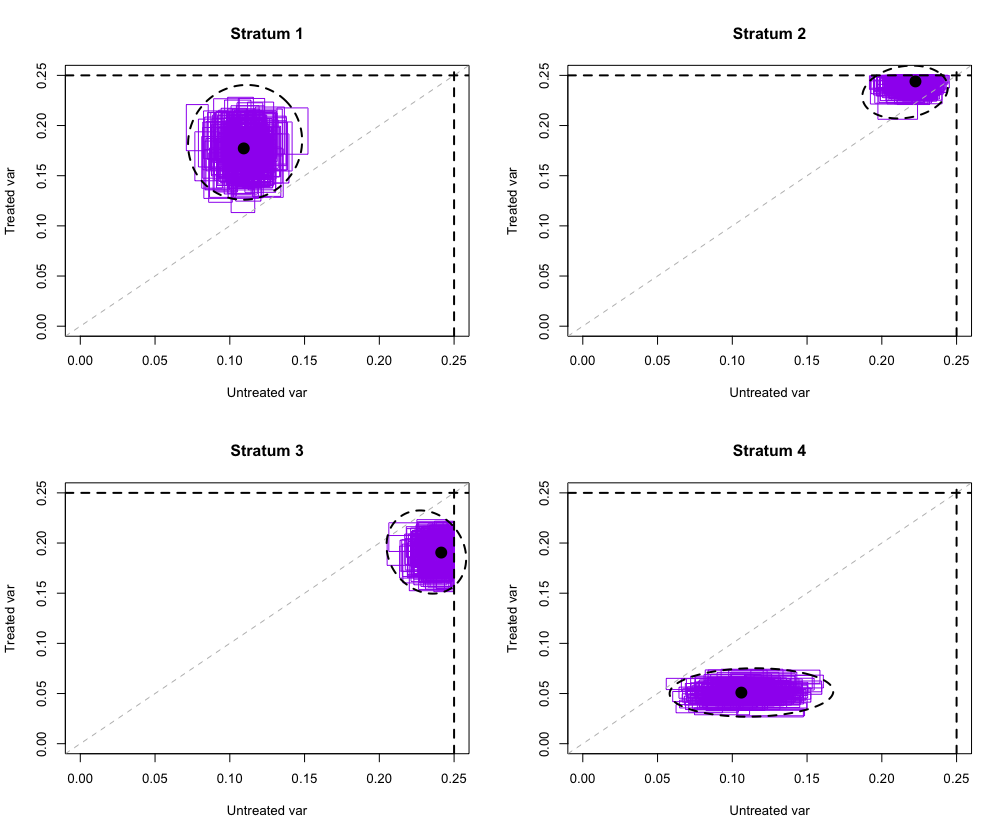}
\caption{\label{fig:rectplot} Simulated example of confidence regions in four strata under $\Gamma = 1.2$. }
\end{figure}

The objective is convex in $n_{rkt}, n_{rkc}$ and affine (and thus concave) in $\sigma^2_{k}(1), \sigma^2_{k}(0)$. Now, having obtained convex constraints, we can invoke Von Neumann's minimax theorem \citep{neumann1928theory} to switch the order of the minimization and maximization. Hence, the solution to Problem \ref{optProb} is equivalent to the solution of 

\begin{align*}
\max_{\sigma_{k}^2(1), \sigma_{k}^2(0)} \min_{n_{rkt}, n_{rkc}}  & \hspace{5mm}  \sum_k w_k \left( \sigma_k^2(1) \left(\frac{1}{n_{rkt}}- \frac{1}{\tilde n_{rkt}}\right) + \sigma_k^2(0) \left(\frac{1}{n_{rkc}}- \frac{1}{\tilde n_{rkc}}\right)\right) \\ \text{subject to} & \hspace{5mm} (\sigma_k^2(1), \sigma_k^2(0)) \in \mathcal{A}_k^{(\Gamma)}, k = 1, \dots, K\\ & \hspace{5mm} \sum_k n_{rkt} + n_{rkc} = n_r \,.
\end{align*}
But the inner problem has an explicit solution, given by
\[ n_{rkt} = n_r\frac{ \sqrt{w_k}  \sigma_k(1)}{\sum_k \sqrt{w_k} (\sigma_k(1) + \sigma_k(0))}, \hspace{5mm} n_{rkc} = n_r\frac{\sqrt{w_k} \sigma_k(0)}{\sum_k \sqrt{w_k} (\sigma_k(1) + \sigma_k(0))}\,.\] 
Plugging this in yields the simplified problem 
\begin{equation}\label{simplifiedProb}
\begin{aligned}
\max_{\sigma_{k}^2(1), \sigma_{k}^2(0)}  & \hspace{5mm} \frac{1}{n_r} \left( \sum_k \sqrt{w_k} \left(\sigma_k(1) + \sigma_k(0) \right)\right)^2 -  \left(\sum_k  w_k\left(\frac{\sigma_k^2(1)}{\tilde n_{rkt}} +  \frac{\sigma_k^2(0)}{\tilde n_{rkc}} \right)\right)
\\ \text{subject to} & \hspace{5mm} (\sigma_k^2(1), \sigma_k^2(0)) \in \mathcal{A}_k^{(\Gamma)}, k = 1, \dots, K\,.
\end{aligned}
\end{equation}
Problem \ref{simplifiedProb} is concave. See Appendix \ref{app:concavityProof} for a detailed proof. The solution is non-trivial, owing to the fact that the problem is not DCP-compliant. Nonetheless, a simple projected gradient descent algorithm is guaranteed to converge under very mild conditions given the curvature \citep{iusem2003convergence}. Hence, we can efficiently solve this problem.  

\section{Application to the Data from the Women's Health Initiative}\label{sec:whi}

\subsection{Setup}

To evaluate our methods in practice, we make use of data from the Women's Health Initiative, a 1991 study of the effects of hormone
therapy on postmenopausal women. The study included both a randomized controlled trial and an observational study. A total of 16,608 women were included in the trial, with half randomly selected to take 625 mg of estrogen and 2.5 mg of progestin, and the remainder receiving a placebo. 
A corresponding 53,054 women in the observational component of the WHI were deemed clinically comparable to women in the trial. About a third of these women were using estrogen plus progestin, while the remaining women in the observational study were not using hormone therapy \citep{prentice2005combined}. 

We investigate the effect of the treatment on incidence of coronary heart disease. The data is split into two non-overlapping subsets, which we term the ``gold" and ``silver" datasets. We estimate the probability of treatment for observational units via fitted propensity scores. The data split is the same as the one used in \cite{rosenman2018propensity}. Details on the construction of these data elements can be found in the Appendix, Section \ref{sec:whiPropScores}, while further details about the WHI can be found in the Supplement, Section \ref{sec:furtherDetailsWHI}. 

To choose our subgroups for stratification, we utilize the clinical expertise of researchers in the study's writing group. The trial protocol highlights age as an important subgroup variable to consider \citep{study1998design}, while subsequent work considered a patient's history of cardiovascular disease \citep{roehm2015reappraisal}. 
We also consider Langley scatter, a measure of solar irradiance at each woman's enrollment center, which is not plausibly related to baseline incidence or treatment effect. Langley scatter exhibits no association with the outcome in the observational control population: a Pearson's Chi-squared test yields a p-value of 0.89. The analogous tests for age and history of cardiovascular disease have p-values below $10^{-5}$. 

The age variable has three levels, corresponding to whether a woman was in her fifties, sixties, or seventies. The cardiovascular disease history variable is binary. The Langley scatter variable has five levels, corresponding to strata between 300 and 500 Langleys of irradiance.  We provide brief summaries of these variables in Tables \ref{table:ageVarDist}, \ref{table:CVDVarDist}, and  \ref{table:LangleyVarDist} in Appendix Section \ref{sec:suppWhi}. 

The RCT gold dataset is used to estimate ``gold standard" stratum causal effects. We now suppose that the observational study is being used to design an experiment of size $n_r = 1,000$ units. We compare the estimates from the designed pseudo-experiments against the gold standard estimates under the unweighted $L_2$ loss.


In the design setting, we face the additional challenge of choosing the appropriate value of $\Gamma$. The WHI provides a very rich set of covariates, and our propensity model incorporates more than 50 variables spanning the demographic and clinical domains (see details in Appendix Section \ref{sec:whiPropScores}). Hence, we will run our algorithm at values of $\Gamma = 1.0$ (reflecting no residual confounding) as well as $\Gamma = 1.1, 1.5,$ and $2.0$ (reflecting a modest amount). 

\subsection{Detailed Example: $\Gamma = 1.5$, Fine Stratification}

We show one example in detail, in which we choose $\Gamma = 1.5$ and stratify on all three subgroup variables: age, history of cardiovascular disease, and Langley scatter. The cross-product of these variables yields 30 subgroups, which we suppose are weighted equally. We number these groups from 1 through 30. 

In the top panel of Figure \ref{fig:AllocationPlot}, we show a na\"ive RCT allocation based purely on the pilot estimates of the stratum potential outcome variances from the observational study. In the bottom panel, we show the regret-minimizing allocations. Visually, it is clear that we have heavily shrunk the allocations toward an equally allocated RCT, but there remain some strata where we recommend over- or under-sampling. Note, too, that the shrinkage is not purely reflective of the magnitude of the pilot estimate, since the number of observational units from each stratum and treatment status also influences the width of our confidence regions for each of the pilot estimates. 

\begin{figure}[H]
\centering
\includegraphics[scale = 0.4]{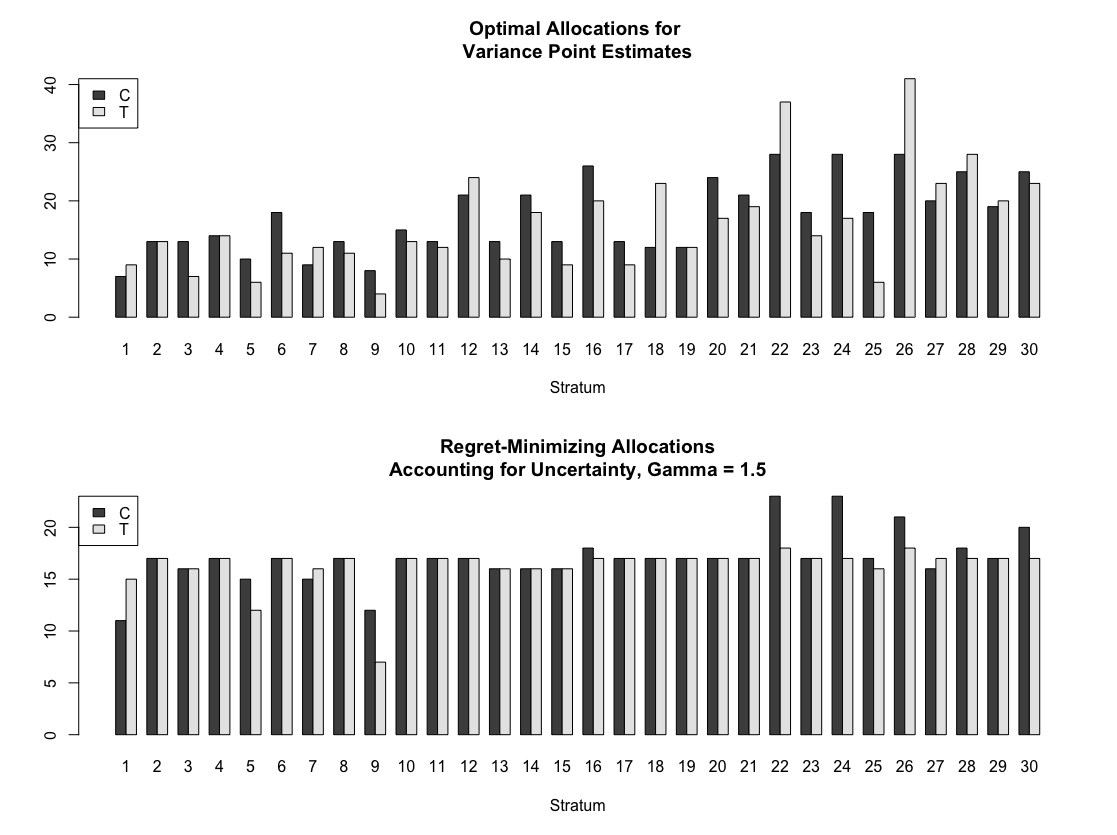}
\caption{\label{fig:AllocationPlot} Allocation of units to strata under na\"ive scheme and regret-minimizing scheme. }
\end{figure}

To investigate the utility of our regret-minimizing allocations, we sample pseudo-experiments of 1,000 units from the RCT silver dataset 1,000 times with replacement. We do so under three designs: equal allocation by strata; na\"ive allocation based on the pilot estimates; and  regret-minimizing allocations. Below, we show the average $L_2$ loss when compared against the gold standard estimates derived from the RCT gold estimate. Results are shown in Figure \ref{fig:avgLossbyDesignType}. Our method yields a modest reduction in average loss (3.6\%) relative to the na\"ive design. It also outperforms the equal design, though by a slimmer margin (1.6\%). This is encouraging -- especially because the design was intended to guard against \emph{worst-case} loss, rather than to optimize average loss. 

\begin{figure}[H]
\centering
\includegraphics[scale = 0.25]{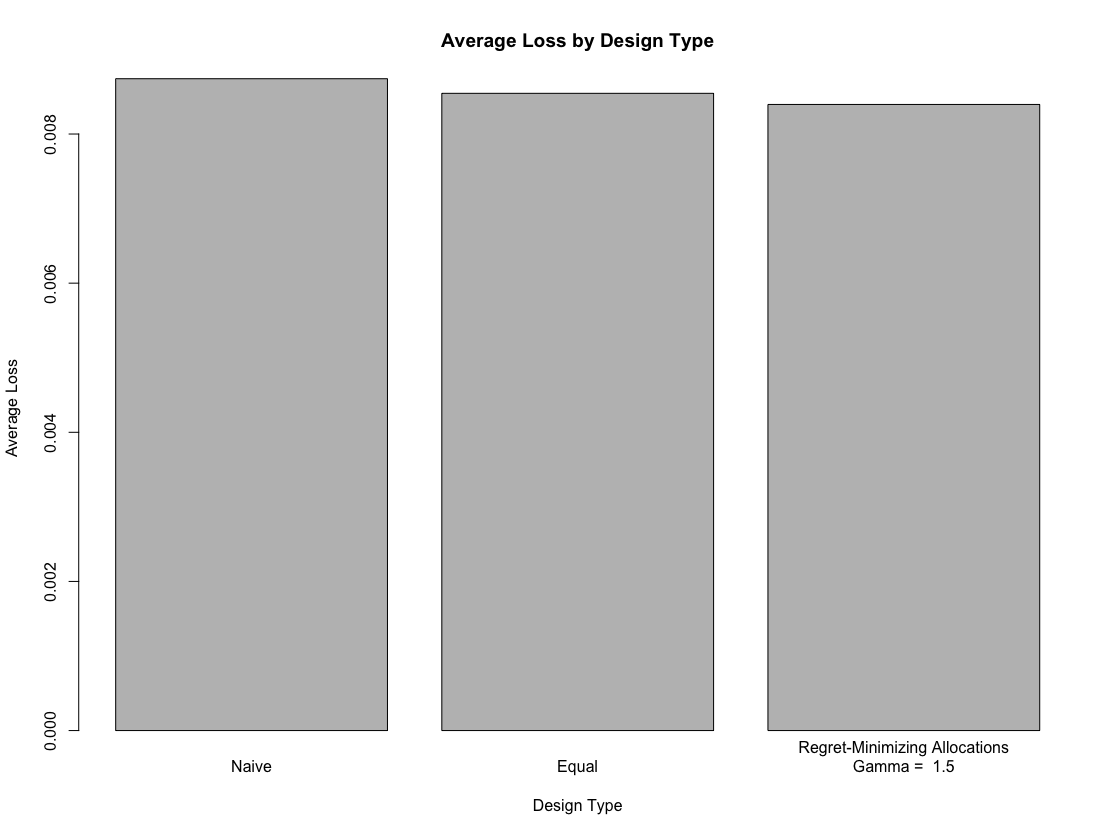}
\caption{\label{fig:avgLossbyDesignType} Average loss over 1,000 resamples of 1,000-units experiments under equal-allocation, na\"ive-allocation, and regret-minimizing allocation designs.}
\end{figure}

\subsection{Performance Over Multiple Conditions}

We now simulate with all possible combinations of the stratification variables. For each choice of a stratification, we select $1,000$ units under equal allocation, na\"ive allocation, and regret-minimizing allocation with $\Gamma = 1.0, 1.1, 1.5$ and $2.0$. We then compute the $L_2$ loss versus the ``gold standard" estimates derived from the RCT gold datasets.

In Table \ref{tab:comparisonToEqualAlloc}, we summarize the loss relative to equal allocation. We see immediately that the entries are all non-positive. This makes some intuitive sense: the objective in Problem \ref{optProb} can always be set to 0 by choosing $n_{rkt} = \tilde n_{rkt}$ and $n_{rkc} = \tilde n_{rkc}$ for all $k$; hence, the algorithm is designed to guarantee that we cannot do worse than allocating equally. By the same token, many of the gains we see are modest, owing to the conservatism of the regret-minimizing approach. Notably, we seem to achieve the greatest gains when we are stratifying only on clinically relevant variables and using a relatively low value of $\Gamma$. We achieve a 5-6\% risk reduction at low values of $\Gamma$ in the fourth row of the table, in which we stratify on the clinically relevant age and cardiovascular disease variables. On the other hand, the algorithm quickly defaults to recommending equal allocation when variables are not clinically relevant. In the third row, in which we stratify only on the irrelevant Langley scatter variable, the starred entries correspond to cases in which the regret-minimizing allocation is equal allocation. 

\begin{table}[h]
\begin{tabular}{lrrrrr}
\multirow{2}{*}{\begin{tabular}[c]{@{}l@{}}Subgroup \\ Variable(s)\end{tabular}} & \multicolumn{1}{l}{\multirow{2}{*}{\begin{tabular}[c]{@{}l@{}}Equal Alloc\\ $L_2$ Loss\end{tabular}}} & \multicolumn{4}{c}{Loss Relative to Equal Allocation}                                                                      \\ \cline{3-6} 
                                                                                 & \multicolumn{1}{r}{}                                                                                   & \multicolumn{1}{r}{\hspace{5mm} $\Gamma = 1$} & \multicolumn{1}{r}{\hspace{5mm} $\Gamma = 1.1$}& \multicolumn{1}{r}{\hspace{5mm} $\Gamma = 1.5$} & \multicolumn{1}{r}{\hspace{5mm} $\Gamma = 2$}  \\ \hline
Age                                                                              & 0.000517                                                                                               & $-$2.0\%                        & $-1.9\%\phantom{*}$                          & $-2.0\%\phantom{*}$                          & $0.0\%\phantom{*}$                         \\
CVD                                                                              & 0.000498                                                                                               & $-$2.3\%                        & $-2.0\%\phantom{*}$                          & $-1.5\%\phantom{*}$                          & $0.0\%\phantom{*}$                         \\
Langley                                                                          & 0.000841                                                                                               & 0.0\%                         & 0.0\%*                           & 0.0\%*                           & 0.0\%*                         \\
Age, CVD                                                                         & 0.001541                                                                                               & $-$5.5\%                        & $-5.6\%\phantom{*}$                          & $-3.8\%\phantom{*}$                          & $-2.3\%\phantom{*}$                        \\
Age, Langley                                                                     & 0.003417                                                                                               & $-$1.6\%                        & $-1.6\%\phantom{*}$                          & $-0.7\%\phantom{*}$                          & $-0.1\%\phantom{*}$                        \\
CVD, Langley                                                                     & 0.002495                                                                                               & $-$1.7\%                        & $-1.2\%\phantom{*}$                          & $-0.8\%\phantom{*}$                          & $-0.2\%\phantom{*}$                        \\
Age, CVD, Langley                                                                & 0.008395                                                                                               & $-$1.9\%                        & $-2.1\%\phantom{*}$                          & $-1.6\%\phantom{*}$                          & $-0.7\%\phantom{*}$                       
\end{tabular}
\caption{\label{tab:comparisonToEqualAlloc} $L_2$ loss comparisons for regret-minimizing allocations relative to equal allocation. For starred entries, the regret-minimizing allocation defaults to equal allocation.}
\end{table}

In Table \ref{tab:comparisonToNaiveAlloc}, we summarize the loss relative to na\"ive allocation. In this case, our method can underperform a na\"ive allocation derived from the observational study pilot variance estimates. This can be seen most clearly in the first row of the table, in which we stratify only on the age variable. Such underperformance is a consequence of the fact that our algorithm is defensive toward underperformance when bias and variance are present in the pilot estimates. However, there are two clear trends in the results. First, when we stratify on a variable that turns out not to be clinically relevant, like Langley scatter, the na\"ive allocation is essentially just recommending an allocation based on noise from the data; as a result, our regret-minimizing allocations uniformly outperform na\"ive allocations. Second, the regret-minimizing allocations tend to outperform the na\"ive allocations as the number of strata grow. We significantly outperform na\"ive allocation in the final row, which corresponds to stratification on all three variables and a total of 30 strata. 
\begin{table}[h]
\begin{tabular}{lrrrrr}
\multirow{2}{*}{\begin{tabular}[c]{@{}l@{}}Subgroup\\ Variable(s)\end{tabular}} & \multirow{2}{*}{\begin{tabular}[c]{@{}l@{}}Na\"ive Alloc\\ $L_2$ Loss\end{tabular}} & \multicolumn{4}{c}{Loss Relative to Na\"ive Allocation} \\ \cline{3-6}                                                                                  & \multicolumn{1}{r}{}                                                                                   & \multicolumn{1}{r}{\hspace{5mm} $\Gamma = 1$} & \multicolumn{1}{r}{\hspace{5mm} $\Gamma = 1.1$}& \multicolumn{1}{r}{\hspace{5mm} $\Gamma = 1.5$} & \multicolumn{1}{r}{\hspace{5mm} $\Gamma = 2$}  \\ \hline
Age                                                                             & 0.000501                                                                       & 1.2\%        & 1.2\%          & 1.1\%          & 3.2\%       \\
CVD                                                                             & 0.000488                                                                       & $-$0.3\%       & 0.0\%          & 0.6\%          & 2.1\%       \\
Langley                                                                         & 0.000852                                                                       & $-$1.1\%       & $-$1.3\%         & $-$1.3\%         & $-$1.3\%      \\
Age, CVD                                                                        & 0.001484                                                                       & $-$1.8\%       & $-$1.9\%         & $-$0.1\%         & 1.5\%       \\
Age, Langley                                                                    & 0.003393                                                                       & $-$0.9\%       & $-$0.9\%         & 0.0\%          & 0.6\%       \\
CVD, Langley                                                                    & 0.002481                                                                       & $-$1.1\%       & $-$0.7\%         & $-$0.3\%         & 0.3\%       \\
Age, CVD, Langley                                                               & 0.008574                                                                       & $-$3.9\%       & $-$4.1\%         & $-$3.6\%         & $-$2.8\%     
\end{tabular}
\caption{\label{tab:comparisonToNaiveAlloc} $L_2$ loss comparisons for regret-minimizing allocations relative to na\"ive allocation.}
\end{table}

While these simulation results show modest performance gains, they are encouraging. A wise analyst would be extremely cautious about designing an RCT exclusively using observational study pilot estimates of stratum variances. Because such pilot estimates can have both bias and variance, relying too heavily upon them might waste resources. Our framework allows data from the observational study to be incorporated into the RCT design while guarding against the possibility of underperforming a default allocation. 

\section{Future Work: General Case}\label{sec:conc}

We briefly discuss challenges in the more general case of $Y_i \in \mathbb{R}$. In keeping with the theme of IPW estimation, we consider estimators of the form 
\begin{equation}\label{eq:ipwEstimatorsGeneralCase}
\begin{aligned}
\hat \sigma_k^2(1) &=  \sum_{i \in \mathcal{O}_k} Y_i^2 \left(\frac{W_i}{p_i}  \right) \bigg/  \sum_{i \in \mathcal{O}_k} \left(\frac{W_i}{p_i}  \right) -  
\left(\sum_{i \in \mathcal{O}_k} Y_i \left(\frac{W_i}{p_i}  \right) \bigg/  \sum_{i \in \mathcal{O}_k} \left(\frac{W_i}{p_i}  \right) \right)^2\\
\hat \sigma_k^2(0) &=  \sum_{i \in \mathcal{O}_k} Y_i^2 \left(\frac{1-W_i}{1-p_i}  \right) \bigg/  \sum_{i \in \mathcal{O}_k} \left(\frac{1-W_i}{1-p_i}  \right) -  
\left(\sum_{i \in \mathcal{O}_k} Y_i \left(\frac{1-W_i}{1-p_i}  \right) \bigg/  \sum_{i \in \mathcal{O}_k} \left(\frac{1-W_i}{1-p_i}  \right) \right)^2\,,
\end{aligned}
\end{equation}
where $p_i$ are the true treatment probabilities. Such estimators are asymptotically unbiased. 

Under the sensitivity model of \cite{zhao2019sensitivity}, we suppose we estimate $p_i$ with fitted propensity scores, $\hat \pi_i$, defined as
\[ \hat \pi_i = \frac{1}{1 + e^{-\hat g(X_i)}}\,. \] 
In the typical setting in which we use logistic regression to estimate the propensity scores, $\hat g(X_i) = \hat \beta^T X_i$. 

We account for the possibility of $\Gamma$-level unmeasured confounding by allowing the true probability $p_i$ to satisfy
\[ p_i \in \left\{ \frac{1}{1 + z_ie^{-\hat g(X_i)}} \hspace{2mm} \bigl\vert \hspace{2mm}\frac{1}{\Gamma} \leq z_i \leq \Gamma\right\}\,, \]
Any affine transformation of our optimization variable will not change the curvature of the problem, so we redefine the problem in terms of the $v_i = p_i^{-1}$, an affine function of the $z_i$. We define two vectors $\bsv_t = \left(v_i\right)_{i: W_i = 1}$ and $\bsv_c = \left(v_i\right)_{i: W_i = 0}$, and analogously define vectors $\boldsymbol{Y_t} = \left(Y_i\right)_{W_i = 1}$ and $\boldsymbol{Y_c} = \left(Y_i \right)_{W_i = 0}$. Now, we can express the equations in \ref{eq:ipwEstimatorsGeneralCase} as quadratic fractional program, e.g. 
\begin{align*}
\hat \sigma_k^2(1) = \frac{\bsv_t^T \boldsymbol{\Theta}_t \bsv_t}{\bsv_t^T \mathbbm{1}\mathbbm{1}^T\bsv_t}, \hspace{3mm} \hat \sigma_k^2(0) = \frac{\bsv_c^T \boldsymbol{\Theta}_c \bsv_c}{\bsv_c^T\mathbbm{1}\mathbbm{1}^T \bsv_c}
\end{align*}
where 
\[ \boldsymbol{\Theta}_t = \boldsymbol{Y_t}^2\mathbbm{1}^T - \boldsymbol{Y_t}\boldsymbol{Y_t}^T\hspace{3mm} \text{ and } \hspace{3mm}\boldsymbol{Y_c}^2\mathbbm{1}^T - \boldsymbol{Y_c}\boldsymbol{Y_c}^T \,.\] 

We have few guarantees on the curvature of the problem: the numerators will be neither convex nor concave in the $\bsv_e$ terms, $e \in \{0, 1\}$, as long as the vectors $\mathbbm{1}, \boldsymbol{Y_t},$ and $\boldsymbol{Y_t}^2$ are linearly independent. The denominators will be convex in the $\bsv_e$ terms. This poses a major challenge. Quadratic fractional programming problems can be solved efficiently in some special cases, but are, in general, NP-hard \citep{phillips1998quadratic}. 

One promising avenue for future work is to apply Dinkelbach's method to transform the quadratic fractional problem to a series of quadratic programming problems \citep{dinkelbach1967nonlinear}. This will not immediately yield a solution because of the indefinite numerator, but it will allow us to make use of considerable recent work on new solution methods in quadratic programming \citep[see e.g.][]{park2017general}. 




\bibliographystyle{apalike} 
\bibliography{rct+odb}      

\begin{thebibliography}{}

\bibitem[Akaike, 1974]{akaike1974new}
Akaike, H. (1974).
\newblock A new look at the statistical model identification.
\newblock {\em IEEE Transactions on Automatic Control}, 19(6):716--723.

\bibitem[Bell, 1982]{bell1982regret}
Bell, D.~E. (1982).
\newblock Regret in decision making under uncertainty.
\newblock {\em Operations Research}, 30(5):961--981.

\bibitem[Boyd et~al., 2004]{boyd2004convex}
Boyd, S., Boyd, S.~P., and Vandenberghe, L. (2004).
\newblock {\em Convex Optimization}.
\newblock Cambridge University Press.

\bibitem[Chin, 2019]{chin2019modern}
Chin, A. (2019).
\newblock {\em Modern statistical approaches for randomized experiments under
  interference}.
\newblock PhD thesis, Stanford University.

\bibitem[Diecidue and Somasundaram, 2017]{diecidue2017regret}
Diecidue, E. and Somasundaram, J. (2017).
\newblock Regret theory: A new foundation.
\newblock {\em Journal of Economic Theory}, 172:88--119.

\bibitem[Dinkelbach, 1967]{dinkelbach1967nonlinear}
Dinkelbach, W. (1967).
\newblock On nonlinear fractional programming.
\newblock {\em Management Science}, 13(7):492--498.

\bibitem[Graziano and Raulin, 1993]{graziano1993research}
Graziano, A.~M. and Raulin, M.~L. (1993).
\newblock {\em Research methods: A process of inquiry}.
\newblock HarperCollins College Publishers.

\bibitem[Hartman et~al., 2015]{hartman2015sate}
Hartman, E., Grieve, R., Ramsahai, R., and Sekhon, J.~S. (2015).
\newblock From {SATE} to {PATT}: Combining experimental with observational
  studies to estimate population treatment effects.
\newblock {\em Journal of the Royal Statistical Society: Series A (Statistics
  in Society)}, 10:1111.

\bibitem[Hays et~al., 2003]{hays2003women}
Hays, J., Hunt, J.~R., Hubbell, F.~A., Anderson, G.~L., Limacher, M., Allen,
  C., and Rossouw, J.~E. (2003).
\newblock The {W}omen's {H}ealth {I}nitiative recruitment methods and results.
\newblock {\em Annals of Epidemiology}, 13(9):S18--S77.

\bibitem[Hill, 2011]{hill2011bayesian}
Hill, J.~L. (2011).
\newblock Bayesian nonparametric modeling for causal inference.
\newblock {\em Journal of Computational and Graphical Statistics},
  20(1):217--240.

\bibitem[Imbens and Rubin, 2015]{Imbens:2015:CIS:2764565}
Imbens, G.~W. and Rubin, D.~B. (2015).
\newblock {\em Causal Inference for Statistics, Social, and Biomedical
  Sciences: An Introduction}.
\newblock Cambridge University Press, New York, NY, USA.

\bibitem[Iusem, 2003]{iusem2003convergence}
Iusem, A.~N. (2003).
\newblock On the convergence properties of the projected gradient method for
  convex optimization.
\newblock {\em Computational \& Applied Mathematics}, 22(1):37--52.

\bibitem[Loomes and Sugden, 1982]{loomes1982regret}
Loomes, G. and Sugden, R. (1982).
\newblock Regret theory: An alternative theory of rational choice under
  uncertainty.
\newblock {\em The Economic Journal}, 92(368):805--824.

\bibitem[Park and Boyd, 2017]{park2017general}
Park, J. and Boyd, S. (2017).
\newblock General heuristics for nonconvex quadratically constrained quadratic
  programming.
\newblock {\em arXiv preprint arXiv:1703.07870}.

\bibitem[Phillips, 2001]{phillips1998quadratic}
Phillips, A.~T. (2001).
\newblock Quadratic fractional programming: Dinkelbach's method.
\newblock In {\em Encyclopedia of Optimization}, volume~4.

\bibitem[Prentice et~al., 2005]{prentice2005combined}
Prentice, R.~L., Langer, R., Stefanick, M.~L., Howard, B.~V., Pettinger, M.,
  Anderson, G., Barad, D., Curb, J.~D., Kotchen, J., Kuller, L., et~al. (2005).
\newblock Combined postmenopausal hormone therapy and cardiovascular disease:
  Toward resolving the discrepancy between observational studies and the
  {W}omen's {H}ealth {I}nitiative clinical trial.
\newblock {\em American Journal of Epidemiology}, 162(5):404--414.

\bibitem[Roehm, 2015]{roehm2015reappraisal}
Roehm, E. (2015).
\newblock A reappraisal of {W}omen's {H}ealth {I}nitiative estrogen-alone
  trial: long-term outcomes in women 50--59 years of age.
\newblock {\em Obstetrics and Gynecology International}, 2015.

\bibitem[Rosenbaum, 2009]{rosenbaum2009design}
Rosenbaum, P. (2009).
\newblock {\em Design of Observational Studies}.
\newblock Springer Series in Statistics. Springer, New York.

\bibitem[Rosenbaum, 1987]{rosenbaum1987sensitivity}
Rosenbaum, P.~R. (1987).
\newblock Sensitivity analysis for certain permutation inferences in matched
  observational studies.
\newblock {\em Biometrika}, 74(1):13--26.

\bibitem[Rosenman et~al., 2018]{rosenman2018propensity}
Rosenman, E., Owen, A.~B., Baiocchi, M., and Banack, H. (2018).
\newblock Propensity score methods for merging observational and experimental
  datasets.
\newblock {\em arXiv preprint arXiv:1804.07863}.

\bibitem[Rubin, 1974]{rubin1974estimating}
Rubin, D.~B. (1974).
\newblock Estimating causal effects of treatments in randomized and
  nonrandomized studies.
\newblock {\em Journal of Educational Psychology}, 66(5):688.

\bibitem[Tan, 2006]{tan2006distributional}
Tan, Z. (2006).
\newblock A distributional approach for causal inference using propensity
  scores.
\newblock {\em Journal of the American Statistical Association},
  101(476):1619--1637.

\bibitem[Van~der Vaart, 2000]{van2000asymptotic}
Van~der Vaart, A.~W. (2000).
\newblock {\em Asymptotic Statistics}, volume~3.
\newblock Cambridge University Press.

\bibitem[VanderWeele and Robins, 2012]{vanderweele2012stochastic}
VanderWeele, T.~J. and Robins, J.~M. (2012).
\newblock Stochastic counterfactuals and stochastic sufficient causes.
\newblock {\em Statistica Sinica}, 22(1):379.

\bibitem[Von~Neumann, 1928]{neumann1928theory}
Von~Neumann, J. (1928).
\newblock On game theory.
\newblock {\em Proceedings of the Academy of Sciences}, 100(1):295--320.

\bibitem[Wager and Athey, 2018]{wager2018estimation}
Wager, S. and Athey, S. (2018).
\newblock Estimation and inference of heterogeneous treatment effects using
  random forests.
\newblock {\em Journal of the American Statistical Association},
  113(523):1228--1242.

\bibitem[Weyl, 1912]{weyl1912asymptotic}
Weyl, H. (1912).
\newblock The asymptotic distribution law for the eigenvalues of linear partial
  differential equations (with applications to the theory of black body
  radiation).
\newblock {\em Mathematical Annals}, 71(1):441--479.

\bibitem[{Writing Group for the {W}omen's {H}ealth {I}nitiative Investigators},
  1998]{study1998design}
{Writing Group for the {W}omen's {H}ealth {I}nitiative Investigators} (1998).
\newblock Design of the {W}omen's {H}ealth {I}nitiative clinical trial and
  observational study.
\newblock {\em Controlled Clinical Trials}, 19(1):61--109.

\bibitem[{Writing Group for the {W}omen's {H}ealth {I}nitiative Investigators},
  2002]{writing2002risks}
{Writing Group for the {W}omen's {H}ealth {I}nitiative Investigators} (2002).
\newblock Risks and benefits of estrogen plus progestin in healthy
  postmenopausal women: Principal results from the {W}omen's {H}ealth
  {I}nitiative randomized controlled trial.
\newblock {\em Journal of the American Medical Association}, 288(3):321--333.

\bibitem[Zhao et~al., 2019]{zhao2019sensitivity}
Zhao, Q., Small, D.~S., and Bhattacharya, B.~B. (2019).
\newblock Sensitivity analysis for inverse probability weighting estimators via
  the percentile bootstrap.
\newblock {\em Journal of the Royal Statistical Society: Series B (Statistical
  Methodology)}, 81(4):735--761.

\end{thebibliography}

\appendix
\section{Appendix}

\subsection{Further Details about the Women's Health Initiative}\label{sec:furtherDetailsWHI}

In this section we evaluate our estimators on data from the Women's Health Initiative 
to estimate the effect of hormone therapy on coronary heart disease. 
The Women's Health Initiative is a study of postmenopausal women in the United States, consisting of randomized controlled trial and observational study components with 161,808 total women enrolled \citep{prentice2005combined}. Eligibility and recruitment data for the WHI can be found in the early results papers \citep{hays2003women, writing2002risks}. Participants were women between 50 and 79 years old at baseline, who had a predicted survival of at least three years and were unlikely to leave their current geographic area for three years. 

Women with a uterus who met various safety, adherence, and retention criteria were eligible for a combined hormone therapy trial. A total of 16,608 women were included in the trial, with 8,506 women randomized to take 625 milligrams of estrogen and 2.5 milligrams of progestin, and the remainder receiving a placebo. A corresponding 53,054 women in the observational component of the Women's Health Initiative had an intact uterus and were not using unopposed estrogen at baseline, thus rendering them clinically comparable \citep{prentice2005combined}. About a third of these women were using estrogen plus progestin, while the remaining women in the observational study were not using hormone therapy \citep{prentice2005combined}. 

Participants received semiannual contacts and annual in-clinic visits for the collection of information about outcomes. 
Disease events, including CHD, were first self-reported and later adjudicated by physicians. We focus on outcomes during the initial phase of the study, which extended for an average of 8.16 years of follow-up in the randomized controlled trial and 7.96 years in the observational study. 


The overall rate of coronary heart disease in the trial was 3.7\% in the treated group (314 cases among 8,472 women reporting) versus 3.3\% (269 cases among 8,065 women reporting) for women not randomized to estrogen and progestin. In the observational study, the corresponding rates were 1.6\% among treated women (706 out of 17,457 women reporting) and 3.1\% among control women (1,108 out of 35,408 women reporting). Our methodology compares means and not survival curves. In the initial follow-up period, death rates were relatively low in both the observational study (6.4\%) and the randomized trial (5.7\%). Hence, we do not correct for the possibility of these deaths censoring coronary heart disease events. 

\subsection{Propensity Score Construction, Covariate Balance, and Gold Standard Effects }\label{sec:whiPropScores}

The Women's Health Initiative researchers collected a rich set of covariates about the participants in the study. For the purposes of computational speed, we narrow to a set of 684 variables, spanning demographics, medical history, diet, physical measurements, and psychosocial data collected at baseline. 

The most meaningful measure of covariate imbalance can be found by looking at clinically relevant factors. 
\cite{prentice2005combined} identified factors that are correlated with CHD. They found that hormone therapy users in the observational study were more likely to be Caucasian or Asian/Pacific Islander, less likely to be overweight, and more likely to have a college degree. These imbalances strongly suggest that applying a na\"ive differencing estimate to the observational data will yield an unfairly rosy view of the effect of hormone therapy on CHD. 

To generate our estimators for this dataset, we need a propensity model $e(\bsx)$ to map the observed covariates to an estimated probability of receiving the treatment in the observational study. We used a logistic regression to generate an expressive model while limiting overfit. A forward stepping algorithm was first applied to the observational dataset to put an ordering on the variables. All 684 baseline covariates were provided as candidates to a logistic regression predicting the treatment indicator, and variables were automatically added, one at a time, based on which addition most reduced Akaike's Information Criterion \citep{akaike1974new}. 

Using this ordering, models containing from one to 120 variables were generated. Model fit was assessed via the area under the Receiver Operator Characteristic curve. At each model size, the area under the curve was computed first for the nominal model and then computed again using a ten-fold cross-validation. This procedure generated the curves seen in Figure \ref{fig:cvROCAUC}. Notably, we observe that the predictive power rises rapidly with the addition of the first twenty variables to the logistic regression model, but slows dramatically thereafter. There is also very little evidence of overfit, as the nominal area under the durve only very slightly outpaces the cross-validated area under the curve, even in models with 100 or more variables. This is likely a consequence of the large number of observations in the observational dataset. 

\begin{figure}[h!]
\caption{\label{fig:cvROCAUC} Nominal and cross-validated receiver operator characteristic area under curve for propensity models with different numbers of variables}
\centering
\includegraphics[width = \textwidth]{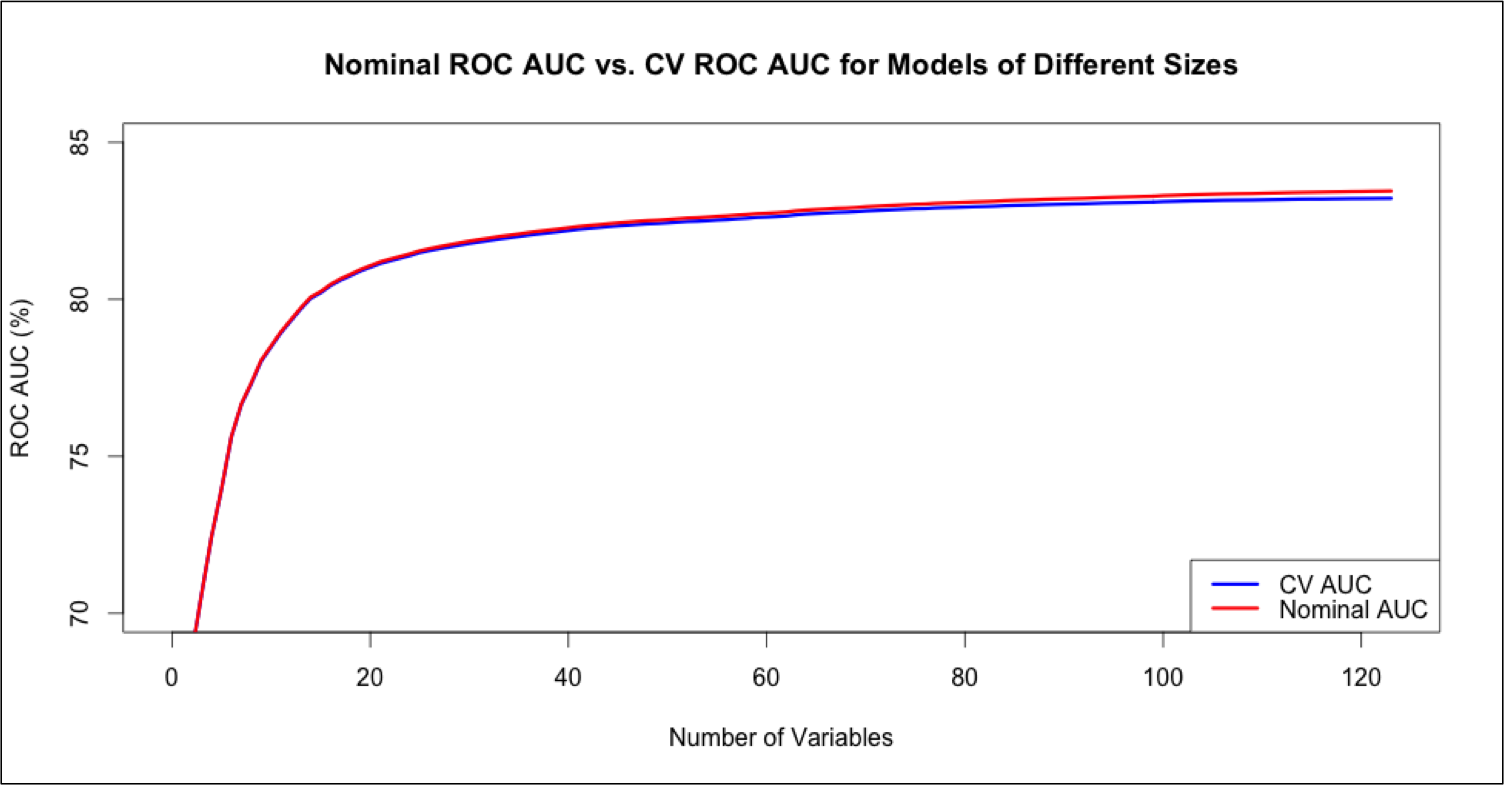}
\end{figure}

We next applied a heuristic threshold, selecting the largest model such that the most recent variable addition increased the cross-validated area under the curve by at least one basis point (0.01\%). This yields a model containing 53 variables, with area under the receiver operator characteristic curve of 82.49\%, or about 1\% lower than a model containing all 684 covariates. 
As our goal is to obtain an association between $e(\bsx_i)$ and $W_i$, and additional variables beyond
the 53rd do not materially improve this association, omission of the additional variables seems warranted. 


Matching on the propensity score should reduce imbalances on clinically relevant covariates. To evaluate this effect, we use standardized differences (as advocated by \cite{rosenbaum2009design}).
Let $\bar x_{tj}$ and $\bar x_{cj}$ be the treated and control group averages for continuous covariate $j$ in the ODB before matching and let $\hat\sigma^2_{tj}$ and $\hat\sigma^2_{cj}$
be the sample variances within those two groups.
Let $\bar x_{tjk}$ and $\bar x_{cjk}$ be those averages taken over subjects $i\in\odb_k$
and define post-stratification averages as $\tilde x_{tj} = \sum_kn_{ok}\bar x_{tjk}/n_o$
and $\tilde x_{cj} = \sum_kn_{ok}\bar x_{cjk}/n_o$.
These are weighted averages of $x_{ij}$ with greater weight put on observations 
from treatment conditions that are underrepresented in their own strata.
Rosenbaum's standardized differences 
for the original and reweighted data are
$$
\overline\sd_j = \frac{\bar x_{tj}-\bar x_{cj}}{\sqrt{\frac12(\hat\sigma^2_{tj}+\hat\sigma^2_{cj})}}
\quad\text{and}\quad
\wt\sd_j = \frac{\tilde x_{tj}-\tilde x_{cj}}{\sqrt{\frac12(\hat\sigma^2_{tj}+\hat\sigma^2_{cj})}},
$$
respectively. These quantities measure the practical significance of the
imbalance between groups unlike $t$-statistics which have a
standard error in the denominator. Note that Rosenbaum uses the same denominator
in both weighted and unweighted standardized differences.

We considered ten equal-width propensity score strata to evaluate the standardized differences between treated and control on risk factors listed in \cite{prentice2005combined}, before and after adjusting for the propensity score. With the exception of the physical functioning score, all of these covariates were included in the propensity model. Imbalance measures for the continuous covariates can be found in Table \ref{tab:balanceCont}. As we can see, the stratification procedure reduces all standardized differences to less than 0.05 in absolute value, representing very good matches between the populations. 

For categorical variables, the stratification procedure similarly reweights individual women, such that the effective proportion of women in each category changes after stratifying on the propensity score. Standardized differences can also be computed for categorical variables, using the procedure described in Graziano et al.\cite{graziano1993research} We achieve similar balance on two significant categorical variables -- ethnicity and smoking status -- in Tables \ref{tab:balanceRace} and \ref{tab:balanceSmoking}. 

\begin{table}[h!]
\caption{\label{tab:balanceCont} Standardized differences (SD) between treated and control populations in the observational dataset, before and after stratification on the propensity score, for clinical risk factors for coronary heart disease.\newline}
\centering
\begin{tabular}{lcccccc}
\toprule
\textbf{}                                                                      & \multicolumn{3}{c}{\textbf{Before Stratifying}}                                                            & \multicolumn{3}{c}{\textbf{After Stratifying}}                                                             \\ 
                                                              & \textbf{Test} & \textbf{Ctrl} & \textbf{\begin{tabular}[c]{@{}l@{}}SD\end{tabular}} & \textbf{Test} & \textbf{Ctrl} & \textbf{\begin{tabular}[c]{@{}l@{}}SD\end{tabular}} \\ 
\midrule
\textbf{Age}                                                                   & 60.78         & 64.72            & $-0.56$                                                                       & 63.06         & 63.33            & $-0.04$                                                                       \\ 
\textbf{BMI}                                                                   & 25.55         & 27.11            & $-0.25$                                                                       & 26.71         & 26.62            & $\phm0.00$                                                                        \\ 
\textbf{\begin{tabular}[c]{@{}l@{}}Physical functioning \end{tabular}} & 85.23         & 79.58            & $\phm0.26$                                                                        & 81.15         & 81.23            & $\phm0.03$                                                                       \\ 
\textbf{Age at menopause}                                                      & 50.49         & 50.19            & $\phm0.06$                                                                        & 50.35         & 50.33            & $\phm0.02$                                                                        \\ 
\bottomrule
\end{tabular}
\end{table}

\begin{table}[h!]
\caption{\label{tab:balanceRace} Standardized differences (SD) between treated and control populations in the observational database, before and after stratification on the propensity score, for ethnicity category.}
\begin{tabular}{lllllllll}
\toprule
                                                                                           &         & \textbf{White} & \textbf{Black} & \textbf{Latino} & \textbf{AAPI} & \textbf{\begin{tabular}[c]{@{}l@{}}Native \\ American\end{tabular}} & \textbf{\begin{tabular}[c]{@{}l@{}}Missing/\\ Other\end{tabular}} & \textbf{\begin{tabular}[c]{@{}l@{}}SD \end{tabular}} \\ \hline
\multirow{2}{*}{\textbf{\begin{tabular}[c]{@{}l@{}}Before \\ Stratifying \end{tabular}}} & Treated    & 89.0\%         & 2.7\%          & 2.9\%           & 4.0\%         & 0.2\%                                                               & 1.1\%                                                             & \multirow{2}{*}{0.26}                                                \\ 
                                                                                           & Control & 83.1\%         & 8.1\%          & 3.9\%           & 2.8\%         & 0.4\%                                                               & 1.5\%                                                             &                                                                      \\ \midrule
\multirow{2}{*}{\textbf{\begin{tabular}[c]{@{}l@{}}After \\ Stratifying\end{tabular}}}  & Treated    & 83.4\%         & 6.9\%          & 4.3\%           & 3.6\%         & 0.5\%                                                               & 1.4\%                                                             & \multirow{2}{*}{0.05}                                                \\
                                                                                           & Control & 84.8\%         & 6.4\%          & 3.6\%           & 3.4\%         & 0.4\%                                                               & 1.4\%                                                             &                                                                      \\ 
                                                                                           \bottomrule
\end{tabular}
\end{table}

\begin{table}[h!]
\caption{\label{tab:balanceSmoking} Standardized differences (SD) between treated and control populations in the observational database, before and after stratification on the propensity score, for smoking category.}
\begin{tabular}{llllll}
\toprule
                                                                                           &         & \textbf{\begin{tabular}[c]{@{}l@{}}Never \\ Smoked\end{tabular}} & \textbf{\begin{tabular}[c]{@{}l@{}}Past \\ Smoker\end{tabular}} & \textbf{\begin{tabular}[c]{@{}l@{}}Current \\ Smoker\end{tabular}} & \textbf{\begin{tabular}[c]{@{}l@{}}SD \end{tabular}} \\ \hline
\multirow{2}{*}{\textbf{\begin{tabular}[c]{@{}l@{}}Before \\ Stratifying\end{tabular}}} & Treated    & 48.7\%                & 46.2\%               & 5.1\%                   & \multirow{2}{*}{0.11}                                                       \\
                                                                                           & Control & 52.3\%                & 41.1\%               & 6.6\%                   &                                                                             \\ \hline
\multirow{2}{*}{\textbf{\begin{tabular}[c]{@{}l@{}}After\\ Stratifying\end{tabular}}}   & Treated    & 50.9\%                & 42.5\%               & 6.6\%                   & \multirow{2}{*}{0.01}                                                       \\ 
                                                                                           & Control & 51.0\%                & 42.7\%               & 6.3\%                   &                                                                             \\ \hline
\end{tabular}
\end{table}

Lastly, consider estimation of the ``gold standard" causal effect. We randomly partition the randomized trial data into two subsets of equal size, such that each contains the same number of treated and control women. We select one of these subsets and refer to it as our ``gold" dataset, to be used for estimating the true causal effect. The remaining subset is referred to as the ``silver" dataset, and is used for evaluating our estimators. 

Because of the randomization, we find that treated and control are already well balanced on the coronary heart disease risk factors in the gold dataset, as summarized in Tables \ref{tab:balanceContRCT}, \ref{tab:balanceRaceRCT}, and \ref{tab:balanceSmokingRCT}.

\begin{table}[h!]
\caption{\label{tab:balanceContRCT} Standardized differences (SD) between treated and control populations in RCT gold dataset, for clinical risk factors for coronary heart disease.}
\begin{tabular}{lllr}
\toprule
\textbf{Variable}                                                             & \textbf{Treated} & \textbf{Control} & \textbf{SD} \\ \hline
\textbf{Age}                                                                  & 63.24         & 63.41            & $-$0.02                                                                \\ 
\textbf{BMI}                                                                  & 28.33         & 28.38            & $-$0.01                                                                 \\ 
\textbf{Physical functioning} & 80.97         & 81.11            & $-$0.01                                                                \\ 
\textbf{Age at menopause}                                                     & 44.97         & 46.33            & $-$0.09                                                                \\ \bottomrule
\end{tabular}
\end{table}

\begin{table}[h!]
\caption{\label{tab:balanceRaceRCT} Standardized differences (SD) between treated and control populations in RCT gold dataset, for ethnicity category.}
\begin{tabular}{llllllll}
\toprule
                 & \textbf{White} & \textbf{Black} & \textbf{Latino} & \textbf{AAPI} & \textbf{\begin{tabular}[c]{@{}l@{}}Native\\ American\end{tabular}} & \textbf{\begin{tabular}[c]{@{}l@{}}Missing/\\ Other\end{tabular}} & \textbf{\begin{tabular}[c]{@{}l@{}}SD \end{tabular}} \\ \hline
\textbf{Treated}    & 84.1\%         & 6.5\%          & 5.5\%           & 2.1\%         & 0.26\%                                                             & 1.6\%                                                             & \multirow{2}{*}{0.05}                                                \\ 
\textbf{Control} & 84.6\%         & 6.8\%          & 5.1\%           & 1.9\%         & 0.40\%                                                             & 1.2\%                                                             &                                                                      \\ \bottomrule
\end{tabular}
\end{table}

\begin{table}[h!]
\caption{\label{tab:balanceSmokingRCT} Standardized differences (SD) between treated and control populations in RCT gold dataset, for smoking category.}
\begin{tabular}{lllll}
\toprule
& \textbf{\begin{tabular}[c]{@{}l@{}}Never  Smoked\end{tabular}} & \textbf{\begin{tabular}[c]{@{}l@{}}Past Smoker\end{tabular}} & \textbf{\begin{tabular}[c]{@{}l@{}}Current Smoker\end{tabular}} & \textbf{\begin{tabular}[c]{@{}l@{}} SD \end{tabular}}\\ \hline
\textbf{Treated}    & 50.1\%         & 38.7\%          & 11.2\%   & \multirow{2}{*}{0.03}                                                \\ 
\textbf{Control} & 50.6\%         & 39.1\%          & 10.2\%                                                      &                                                                      \\ \bottomrule
\end{tabular}
\end{table}

\subsection{Stratification Variable Distributions} \label{sec:suppWhi}

In Tables \ref{table:ageVarDist}, \ref{table:CVDVarDist}, and \ref{table:LangleyVarDist},  we provide the distributions for the variables with which we stratify in the main text.

\begin{table}[h!]
\begin{tabular}{cccc}
\hline
\textbf{Age} & \textbf{\begin{tabular}{c} Observational \\ Study \end{tabular}}& \textbf{\begin{tabular}[c]{@{}l@{}}RCT\end{tabular}} &\textbf{\begin{tabular}{c} RCT ``Silver" \\ Dataset \end{tabular}}\\ \hline
\textbf{50-59}     & 17,447 (33.0\%)                                                                        & 5,491 (33.2\%)                                                      & 2,806 (33.9\%)                                                                         \\ 
\textbf{60-69}     & 23,030 (43.6\%)                                                                        & 7,473 (45.2\%)                                                      & 3,689 (44.6\%)                                                                         \\ 
\textbf{70-79}     & 12,388 (23.4\%)                                                                        & 3,573 (21.2\%)                                                      & 1,774 (21.5\%)                                                                         \\ \hline
\end{tabular}
\caption{\label{table:ageVarDist} Distribution of age variable values in the observational study, RCT, and RCT ``silver" datasets.}
\end{table}

\begin{table}[h]
\begin{tabular}{cccc}
\hline
\textbf{\begin{tabular}{c} History of \\ Cardiovascular \\ Disease \end{tabular}} & \textbf{\begin{tabular}{c} Observational \\ Study \end{tabular}}& \textbf{\begin{tabular}[c]{@{}l@{}}RCT\end{tabular}} &\textbf{\begin{tabular}{c} RCT ``Silver" \\ Dataset \end{tabular}}\\ \hline
\textbf{Yes}     & 8,709 (16.5\%)                                                                        & 1,828 (11.1\%)                                                      & 900 (10.9\%)                                                                         \\
\textbf{No}     & 44,156 (83.5\%)                                                                        & 14,709 (88.9\%)                                                      & 7,369 (89.1\%)                                                                         \\ \hline
\end{tabular}
\caption{\label{table:CVDVarDist} Distribution of history of cardiovascular disease in the observational study, RCT, and RCT ``silver" datasets.}
\end{table}

\begin{table}[h]
\begin{tabular}{cccc}
\hline
\textbf{\textbf{\begin{tabular}{c} Langley Scatter \\ (g-cal/$\text{cm}^2$) \end{tabular}}} & \textbf{\begin{tabular}{c} Observational \\ Study \end{tabular}}& \textbf{\begin{tabular}[c]{@{}l@{}}RCT\end{tabular}} &\textbf{\begin{tabular}{c} RCT ``Silver" \\ Dataset \end{tabular}}\\ \hline
\textbf{300-325}     & 15,599 (29.5\%)                                                                        & 4,854 (29.4\%)                                                      & 2,411 (29.2\%)                                                                         \\ 
\textbf{350}     & 12,521 (23.7\%)                                                                        & 3,917 (23.7\%)                                                      & 1,935 (23.4\%)                                                                         \\ 
\textbf{375-380}     & 5,841 (11.0\%)                                                                        & 1,858 (11.2\%)                                                      & 934 (11.3\%)                                                                         \\ 
\textbf{400-430}     & 8,216 (15.5\%)                                                                        & 2,585 (15.6\%)                                                      & 1,310 (15.8\%)                                                                         \\ 
\textbf{475-500}     & 10,688 (20.2\%)                                                                        & 3,323 (20.1\%)                                                      & 1,679 (20.3\%)                                                                         \\ \hline
\end{tabular}
\caption{\label{table:LangleyVarDist} Distribution of Langley scatter categories in the observational study, RCT, and RCT ``silver" datasets.}
\end{table}

\section{Proof of Validity of Confidence Regions}\label{app:validConfRegions}

We hew closely to the proofs provided in  \cite{zhao2019sensitivity}. Their primary proofs consider the missing data problem, which is equivalent to estimating the mean of either of the potential outcomes. We begin by providing details of their proof and then show how it can be extended to our case. 

\subsection{Review of Proof in \cite{zhao2019sensitivity}}

The authors define $e(\bsx, y)$ as the probability of treatment given covariates $\bsX = \bsx \in \mathcal{X}$ and outcome $Y = y \in \mathbb{R}$ and compare it against the marginal treatment probability $e(\bsx)$. They use $A$ rather than $W$ to denote a treatment indicator, so in keeping with their notation: 
\[ e(\bsx, y) = P_0(A = 1 \mid \bsX = \bsx, Y = y) \hspace{3mm} \text{ and } e(\bsx, y) = P_0(A = 1 \mid \bsX = \bsx) \]
Then, for any choice of $\Lambda > 1$, they define a collection of sensitivity models
\[ \mathcal{E}(\Lambda) = \left\{ 0 \leq e(\bsx, y) \leq 1 \mid \frac{1}{\Lambda} \leq \OR(e(\bsx, y), e(\bsx)), \text{ for all $\bsx, y$} \right\} \] 
where $\OR(p_1, p_2) = [p_1/(1-p_1)]/[p_2/(1-p_2)]$ is the odds ratio. This model was originally introduced by \cite{tan2006distributional}. Per proposition 7.1, it is related to the widely used Rosenbaum sensitivity model. In keeping with that model, we use $\Gamma$ rather than $\Lambda$ to denote our sensitivity parameter in the text, but retain the notation $\Lambda$ throughout this proof. 

Via remark 3.2, Zhao and co-authors reparameterize the problem such that each model in $\mathcal{E}(\Lambda)$ corresponds to a choice of $h(\bsx, y)$, the logit-scale difference of the observed probability $e(\bsx)$ and the complete data selection probability $e(\bsx, y)$. So we can alternatively write: 
\[ \mathcal{E}(\Lambda) = \left\{ e^{(h)}(\bsx, y) \mid h \in \mathcal{H}(\lambda) \right\} \] 
where $\lambda = \log(\Lambda)$ and $\mathcal{H}(\lambda) = \{h: \mathcal{X} \times \mathbb{R} \mid ||h||_{\infty} \leq \lambda\}$. In words: every choice of $h \in \mathcal{H}(\lambda)$ defines, at each possible value of $\bsX$ and $Y$, a discrepancy between $e(\bsx)$ and $e(\bsx, y)$. The choice of $\mathcal{H}(\lambda)$ bounds the maximum of those discrepancies. So, as $\Lambda$ grows, we are allowing for greater and greater discrepancies in these probabilities. 

For each choice of $h$, they define a ``shifted estimand,"
\[ \mu^{(h)} = \left( \E \left( \frac{AY}{e^{(h)}(\bsx, y)} \right)\right)^{-1} \E \left( \frac{AY}{e^{(h)}(\bsx, y)} \right) \] 
where $A$ is the treatment indicator and the expectation is over the joint distribution of $\bsX, Y, A$. The corresponding ``shifted estimator" is given by
\[ \hat \mu^{(h)} = \left( \frac{1}{n} \sum_{i = 1}^n \frac{A_i}{\hat e^{(h)}(\bsX_i, Y_i) }\right)^{-1} \frac{1}{n} \sum_{i = 1}^n \frac{A_i Y_i}{\hat e^{(h)}(\bsX_i, Y_i)} \,. \] 
The sum is over a sample of points $(\bsX_i, Y_i, A_i)$ drawn i.i.d. from their joint distribution. The  quantity in the denominators, $\hat e^{(h)}(\bsX_i, Y_i)$, is obtained by estimating $P(A = 1 \mid \bsX = \bsx)$ and then shifting the estimate by $h(\bsx_i, y_i)$ for all units $i$ such that $\bsX_i = \bsx_i$ and $Y_i = y_i$. 

Now, the proof of the validity of their approach proceeds in several stages. 
\begin{enumerate}
    \item First, they consider the case where data-dependent intervals $[L^{(h)}, U^{(h)}]$ are asymptotically guaranteed to contain $\mu^{(h)}$ with $1-\alpha$ probability. They argue that taking $L = \inf_{h \in \mathcal{H}(\lambda)} L^{(h)}$ and $U = \sup_{h \in \mathcal{H(\lambda)}} U^{(h)}$ yields an interval $[L, H]$ with asymptotic $1-\alpha$ coverage for every value of $\mu^{(h)}$ for which $h \in \mathcal{H}(\lambda)$. (Proposition 4.1).
    \item For each choice of $h \in \mathcal{H}(\lambda)$, they establish that the bootstrap is valid (Theorem 4.2).
    \begin{itemize}
        \item First, they use the general theory of Z-estimators to show that $\hat \mu^{(h)}$ and its bootstrap analogue, $\hat {\hat{\mu}}^{(h)}$, are asymptotically normal with the same mean and variance. (Theorem C.1 and Corollary C.2)
        \item Then, they conclude that defining $L_B^{(h)}$ as the $\alpha/2$ bootstrap quantile, they have
        \[ P\left( \mu^{(h)} < L_B^{(h)}\right) \to \frac{\alpha}{2} \] 
        where the expectation is taken under the joint distribution of $\bsX, Y$ and $A$. Analogous results holds for $U_B^{(h)}$, the $1 - \alpha/2$ bootstrap quantile. (Section C.3)
    \end{itemize}
    \item They argue that the quantile and infimum/supremum functions can be interchanged, such that 
    \[ Q_{\alpha/2} \left( \inf_{h \in \mathcal{H}(\lambda)} \hat {\hat{\mu}}^{(h)} \right) \leq \inf_{h \in \mathcal{H}(\lambda)} L^{(h)} \] 
    and
    \[ Q_{1-\alpha/2} \left( \sup_{h \in \mathcal{H}(\lambda)} \hat {\hat{\mu}}^{(h)} \right) \geq \sup_{h \in \mathcal{H}(\lambda)} U^{(h)} \] 
    via Lemma 4.3. 
\end{enumerate}

\subsection{Extension to Design Case}

Our challenge is to extend this argument to the case where our estimand of interest is not a single $\mu$ but rather the pair $(\sigma_k^2(0), \sigma_k^2(1)) = (\mu_k(0)(1-\mu_k(0)), \mu_k(1)(1-\mu_k(1))$. Crucially, we will now have two $h$ functions $h_0$ and $h_1$, corresponding to each of the potential outcomes, but they both lie within $H(\lambda)$. The definition of the shifted estimand under $h$ given above generalizes to the case of two shifted estimands in a straightforward way. We extend Proposition 1 in the following argument. 

\begin{proposition}\label{prop:dataDepIntervals}
Suppose there exists a data-dependent region $\beta_k^{(h_0, h_1)} \in \mathbb{R}^2$ such that
\[ \liminf_{n \to \infty} P\left(\left(\sigma_k^{(h_0)}(0)^2, \sigma_k^{(h_1)}(1)^{2}\right) \in \beta_k^{(h_0, h_1)} \right) \geq 1 - \alpha \] 
holds for every $(h_0, h_1) \in \mathcal{H}(\lambda)\times  \mathcal{H}(\lambda)$, where $\sigma_k^{(h_0)}(e)^2 = \mu_k^{(h_1)}(e) (1 - \mu_k^{(h)}(e))$ for $e \in \{0, 1\}$, and $n$ is the sample size. 
Under these conditions, the set 
\[ \beta_k = \bigcup_{h_0, h_1 \in \mathcal{H}(\lambda)} \beta_k^{(h_0, h_1)}\]
is an asymptotic confidence set of $(\sigma_k^2(0), \sigma_k^2(1))$ with at least $1 - \alpha$ coverage if $h_0, h_1 \in \mathcal H(\lambda)$. 
\end{proposition}
\begin{proof}
This follows from the fact that, by assumption, the true data-generating distribution satisfies in $h_0, h_1 \in \mathcal H(\lambda)$. 
\end{proof}

Next, we must show that the bootstrap is valid in our setting. We adopt the same model and regularity conditions of Theorem 4.2 in \cite{zhao2019sensitivity}. In their proof of Corollary 5.1, the authors show that the pairs $\left(\hat \mu_k^{(h_0)}(0), \hat \mu_k^{(h_1)}(1)\right)$ and $\left(\hat{\hat \mu}_k^{(h_0)}(0), \hat{\hat{\mu}}_k^{(h_1)}(1)\right)$ are both jointly asymptotically normal, with the same limiting distribution. We define the function
\[ f(x, y) = \left(x\cdot(1-x), y \cdot(1-y)\right) \,.\]
We can see that applying $f(\cdot)$ to the tuple of potential outcome means will yield the potential outcome variances, and the same logic holds for applying $f(\cdot)$ to any estimator of the potential outcome means. Moreover, because $f(\cdot)$ is continuously differentiable, we can use the Delta Method to observe immediately that $\left(\hat \sigma_k^{(h_0)}(0)^2, \hat \sigma_k^{(h_1)}(1)^2\right)$ and $\left(\hat{\hat \sigma}_k^{(h_0)}(0)^2, \hat{\hat{\sigma}}_k^{(h_1)}(1)^2\right)$ have the same asymptotic distribution, and thus the bootstrap is valid \citep{van2000asymptotic}. 


Lastly, we generalize Lemma 4.3 to our setting. For each possible bootstrap replicate $b \in \{1, \dots, N\}$ where $N = n^n$, define the quartet of points
\[ \hat{\hat R}_{k, b} = \left\{ \begin{array}{c}  \left(\inf_{h_0 \in \mathcal H(\lambda)} \hat{\hat \sigma}_k^{(h_0)}(0)^2, \inf_{h_1 \in \mathcal H(\lambda)} \hat{\hat \sigma}_k^{(h_1)}(1)^2 \right), \\ 
 \left(\inf_{h_0 \in \mathcal H(\lambda)} \hat{\hat \sigma}_{k,b}^{(h_0)}(0)^2, \sup_{h_1 \in \mathcal H(\lambda)} \hat{\hat \sigma}_{k,b}^{(h_1)}(1)^2 \right), \\
  \left(\sup_{h_0 \in \mathcal H(\lambda)} \hat{\hat \sigma}_{k,b}^{(h_0)}(0)^2, \inf_{h_1 \in \mathcal H(\lambda)} \hat{\hat \sigma}_{k,b}^{(h_1)}(1)^2 \right), \\
   \left(\sup_{h_0 \in \mathcal H(\lambda)} \hat{\hat \sigma}_{k,b}^{(h_0)}(0)^2, \sup_{h_1 \in \mathcal H(\lambda)} \hat{\hat \sigma}_{k,b}^{(h_1)}(1)^2 \right) \end{array} \right\} \,.\] 
In words, $\hat{\hat R}_{k, b}$ contains the vertices of a rectangle in $\mathbb{R}^2$ which defines the extrema of the potential outcome variances consistent with $h_0, h_1 \in \mathcal{H}(\lambda)$. 

Denote as $\Conv(\cdot)$ the standard convex hull operator. Define a related operator, 
\[ \Conv^{\star}\left(S, \mathcal{B}\right) = \Conv \left( \bigcup_{b \in \mathcal{B}}  S_b\right) \] 
which takes in a set $S$ of cardinality $N_S$ as well as a set $\mathcal{B} \subset \{1, \dots, N_S\}$. The function returns the convex hull of the points contained in the entries in $S$ indexed by $\mathcal{B}$.

We choose a set $\mathcal{B}_{\alpha} \subseteq \{1, 2, \dots, N = n^n\}$ such that $|\mathcal{B}_{\alpha}| = (1-\alpha)N$, and we define the set
\[ \mathcal{A}_k =\Conv^{\star}\left(\{\hat{\hat R}_{k, b}\}, \mathcal{B}_{\alpha}\right) \,. \]

\begin{lemma}
The set $\mathcal{A}_k$ is an asymptotically valid confidence set.
\end{lemma}
\begin{proof}
For $1 \leq b \leq N$, where $N = n^n$ is the total number of possible bootstrap samples, we have that for every $h_0, h_1 \in \mathcal{H}(\lambda)$,
\[ \left(\hat {\hat \sigma}_{k,b}^{(h_0)}(0)^2, \hat {\hat \sigma}_{k,b}^{(h_1)}(1)^2\right) \in \Conv(\hat{\hat R}_{k, b}), \hspace{3mm} \text{for all $ 1 \leq b \leq N$,} \] 
Since this holds entrywise, it follows that any set containing a fixed proportion of the sets on the RHS must contain at least that proportion of points on the LHS, and hence 
\[\Conv^{\star} \left(\left\{ \left(\hat {\hat \sigma}_{k,b}^{(h_0)}(0)^2, \hat {\hat \sigma}_{k,b}^{(h_1)}(1)^2\right)\right\}, \mathcal{B}_{\alpha} \right) \subseteq \Conv^{\star} \left( \left\{\Conv(\hat{\hat R}_{k, b}) \right\}, \mathcal{B}_{\alpha} \right) \,.\] 

Since this holds for every $h_0, h_1 \in \mathcal{H}(\lambda)$, we can take the union on the LHS to observe
\[\bigcup_{h_0, h_1 \in \mathcal{H}(\lambda)} \Conv^{\star} \left(\left\{ \left(\hat {\hat \sigma}_{k,b}^{(h_0)}(0)^2, \hat {\hat \sigma}_{k,b}^{(h_1)}(1)^2\right)\right\}, \mathcal{B}_{\alpha} \right) \subseteq \Conv^{\star} \left( \left\{\Conv(\hat{\hat R}_{k, b}) \right\}, \mathcal{B}_{\alpha} \right) \,.\] 
Observe that the RHS is simply $\mathcal{A}_k$, since any ellipse containing the vertices of a rectangle will contain the convex hull of those vertices as well. 

On the LHS, we can make use of our bootstrap validity result to observe 
\[ \liminf_{n \to \infty} P\left( \left(\sigma_k^{(h_0)}(0)^2, \sigma_k^{(h_1)}(1)^{2}\right) \in  \Conv^{\star} \left(\left\{ \left(\hat {\hat \sigma}_{k,b}^{(h_0)}(0)^2, \hat {\hat \sigma}_{k,b}^{(h_1)}(1)^2\right)\right\}, \mathcal{B}_{\alpha} \right)  \right)   \geq 1- \alpha \,.\] 
It follows from Proposition \ref{prop:dataDepIntervals} that the LHS is a valid $1- \alpha$ level confidence region. Hence, the right-hand side must be as well.

To conclude, we observe that our ellipsoid method must necessarily comprise a superset of a convex hull for some choice of $\mathcal{B}_{\alpha}$. Hence, our method will indeed generate valid confidence regions for the potential outcome variances. 

\end{proof}





\section{Proof of Concavity of Minimax Problem}\label{app:concavityProof}
We begin with the unweighted case, and demonstrate concavity by direct computation of the Hessian. Define
\[ f(\sigma_{k}^2(\cdot)) = \frac{1}{n_r} \left( \sum_k \sigma_k(1) + \sigma_k(0) \right)^2 - \left(\sum_k  \frac{\sigma_k^2(1)}{\tilde n_{rkt}} +  \frac{\sigma_k^2(0)}{\tilde n_{rkc}} \right) \] 
The Hessian is given by 
\begin{align*}
\nabla^2 f &= \frac{1}{2n} \left(H + vv^T \right)
\end{align*}
where 
\[ H = \diag\left(- \frac{\sum_j \sigma_j(0) + \sigma_j(1)}{\sigma_k^3(t)} \right)_{k, t} \hspace{5mm} \text{ and } \hspace{5mm} v = \left(\frac{1}{\sigma_1^2(0)}, \frac{1}{\sigma_1^2(1)}, \dots, \frac{1}{\sigma_K^2(0)}, \frac{1}{\sigma_K^2(1)} \right)^T \,.\]

We want to consider the eigenvalues of $H + vv^T$. First, observe that at most one eigenvalue can be nonnegative. This follows from the famed Weyl Inequalities \cite{weyl1912asymptotic}. $H$ has all strictly negative eigenvalues, while $vv^T$, being an outer product, has one positive eigenvalue, $v^T v$, with all other eigenvalues 0. Denoting as $\lambda_i(G)$ the $i^{th}$ largest eigenvalue of matrix $G$, the Weyl Ineqalities tell us that 
\[ \lambda_2(H + vv^T) \leq \lambda_1(H) + \lambda_2(vv^T) = \lambda_1(H) < 0\,. \] 
Hence, only one non-negative eigenvalue is possible. 

Next, we can use the matrix determinant lemma to observe that
\[ \det(H + vv^T) = (1 + v^TH^{-1}v) \det(H) \] 
and direct computation tells us that 
\[ v^TH^{-1}v = -1 \,.\] 
Hence, the determinant is 0, meaning at least one of our eigenvalues must be zero. Combined with our prior result, this means our maximum eigenvalue must be zero and we conclude the Hessian is negative semidefinite. Thus, $f$ is indeed concave. 

Finally, note that the extension to the weighted case is straightforward. We can simply define new variables $\tilde \sigma_k(e) = \sqrt{w_k} \sigma_k(e)$ for $e \in \{0, 1\}$, and then repeat the proof above using the $\tilde \sigma_k(e)$ variables. Since $\sigma_k(e)$ is simply an affine transformation of $\tilde \sigma_k(e)$, concavity in the former follows from concavity in the latter.

\end{document}